\documentclass[a4paper,11pt]{article}
\usepackage{amsfonts}
\usepackage{amsmath}
\usepackage{amsthm}
\usepackage{amssymb}
\usepackage{mathtools}
\usepackage[hyphens]{url}
\usepackage{hyperref}
\usepackage{graphicx}
\usepackage{subcaption}
\usepackage{xspace}
\usepackage{units}
\usepackage[capitalize,noabbrev,nameinlink]{cleveref}
\usepackage{xcolor}
\usepackage{fullpage}
\usepackage{microtype}

\newtheorem{theorem}{Theorem}
\newtheorem{corollary}[theorem]{Corollary}
\newtheorem{proposition}[theorem]{Proposition}
\newtheorem{lemma}[theorem]{Lemma}

\title{Minimizing Visible Edges in Polyhedra\footnote{A two-page extended abstract of this paper appeared in the \emph{Abstracts of the 23rd Thailand-Japan Conference on Discrete and Computational Geometry, Graphs, and Games (TJCDCGGG)}, pp.~70--71, Chiang Mai, 2021.}}

\author{Csaba D. T\'oth\thanks{Department of Mathematics, California State University Northridge, Los Angeles, CA; Department of Computer Science, Tufts University, Medford, MA, USA. {\tt csaba.toth@csun.edu}}
   \and Jorge Urrutia\thanks{Instituto de Matem\'aticas, Universidad Nacional Aut\'onoma de M\'exico, Mexico City, Mexico. {\tt urrutia@matem.unam.mx}}
   \and Giovanni Viglietta\thanks{Department of Computer Science and Engineering, University of Aizu, Tsuruga, Ichimi-cho, 965-8580 Aizuwakamatsu-shi, Fukushima, Japan. {\tt viglietta@gmail.com}}
}

\index{T\'oth, Csaba D.}
\index{Urrutia, Jorge}
\index{Viglietta, Giovanni}

\date{}

\begin{document}
\thispagestyle{empty}
\maketitle


\begin{abstract}
We prove that, given a polyhedron $\mathcal P$ in $\mathbb{R}^3$, every point in $\mathbb R^3$ that does not see any vertex of $\mathcal P$ must see eight or more edges of $\mathcal P$, and this bound is tight. More generally, this remains true if $\mathcal P$ is any finite arrangement of internally disjoint polygons in $\mathbb{R}^3$. We also prove that every point in $\mathbb{R}^3$ can see six or more edges of $\mathcal{P}$ (possibly only the endpoints of some these edges) and every point in the interior of $\mathcal{P}$ can see a positive portion of at least six edges of $\mathcal{P}$. These bounds are also tight.
\end{abstract}

{\footnotesize {\bf Keywords:} edge guard; polyhedron; visibility graph; spherical occlusion diagram}


\section{Introduction}\label{Sec: Intro}

Computer vision, with applications to motion planning, robotics, and machine learning, aims to arrange our physical environment into data structures. Three-dimensional solids are usually discretized, often represented by polyhedral surfaces, and one would like to compute the view from a given (possibly moving) viewpoint. While visibility problems are motivated by three-dimensional applications, most of the algorithmic, combinatorial, and geometric results pertain to the plane~\cite{OR17,urrutia2000art}. Indeed, some of the most basic properties of visibility in the plane do not easily generalize to three and higher dimensions. This paper address one such problem, related to the minimum number of edges of a polyhedron or polygonal scene visible from a viewpoint.

\smallskip\noindent\textbf{Planar versus three-dimensional visibility.}
Suppose that $s$ is a point in the interior of a simple polygon $P$ in the plane. It is easy to prove that $s$ sees at least three vertices and at least three edges of $P$, and both bounds are tight, e.g., when $P$ is a triangle. Indeed, a simple polygon has a triangulation, the point $s$ lies in some triangle $T$ (possibly on the boundary between two triangles), and so $s$ sees all three vertices of $T$, which are vertices of $P$. Consider the edges of $P$ that are at least partially visible to $s$, and project them orthographically onto a circle centered at $s$. Each edge projects to a circular arc strictly shorter than a semicircle, and the arcs jointly cover the entire circle; consequently, at least three edges are visible to $s$. Similarly, suppose $L$ is an arrangement of $n\geq 3$ pairwise noncrossing line segments in the plane, not all in a line, and $s$ is a point in the convex hull of $L$ disjoint from all segments. Then it is easy to prove that $s$ sees at least three segment endpoints, but possibly only two of the segments.

Analogous statements for polyhedra and for arrangements of interior-disjoint polygons in $\mathbb{R}^3$ are not so straightforward. There are well-known constructions~\cite{Cha84,o1987art} for nonconvex polyhedra $\mathcal{P}$ and viewpoints $s$ in $\mathbb{R}^3$ such that $s$ does not see any vertex of $\mathcal{P}$, even if $\mathcal{P}$ is contractible or orthogonal; see \cref{fig:2}(center) for an example. It is not difficult to show that $s$ sees at least three edges of $\mathcal{P}$: In a cross-section of $\mathcal{P}$ with any plane $\alpha$ containing $s$ but in general position with respect to $\mathcal{P}$, the point $s$ sees at least three vertices, which correspond to edges of $\mathcal{P}$. However, this bound is not tight: We improve it to six (\cref{thm:preli1}), and further to eight when $s$ does not see any vertices of $\mathcal{P}$ (\cref{thm:main1}).

\subsection{Main Results}
For clarity, we define polyhedra and visibility in three-space.
\begin{itemize}
\item A \emph{polygon in the plane} is a connected compact 2-dimensional manifold with boundary such that its boundary is the union of finitely many line segments. A \emph{polygon in $\mathbb{R}^3$} is a polygon in some affine plane in $\mathbb{R}^3$. A polygon is \emph{simple} if it is contractible (or, equivalently if its boundary is connected).
\item A \emph{polyhedron} in $\mathbb{R}^3$ is a connected compact 3-dimensional manifold with boundary such that its boundary is the union of finitely many polygons (which are the \emph{facets} of the polyhedron); see~\cref{fig:1:1} for an example.
\end{itemize}

Given a polyhedron $\mathcal{P}$, two points $p,q\in \mathbb{R}^3$ mutually \emph{see} each other if the closed line segment $pq$ is disjoint from the interior of $\mathcal{P}$, or it is disjoint from the exterior of $\mathcal{P}$. Note that if $p$ is in the interior of $\mathcal{P}$, and $q$ in its exterior, then they are not visible to one other; however, according to our definition, if $p$ and $q$ see each other, then $p$ or $q$ could be on the boundary of $\mathcal{P}$, and the segment $pq$ may overlap with a facet of $\mathcal{P}$ or graze some edges of $\mathcal{P}$.

For visibility between a point and an edge, we adopt the notion of weak visibility: A point $p\in \mathbb{R}^3$ \emph{sees} a line segment $s\subset \mathbb{R}^3$ if there exists a point $q\in s$ such that $p$ sees $q$. In particular, a point $p$ sees an edge $e$ of $\mathcal{P}$ if $p$ sees a vertex of $e$.
We also consider a stronger notion: A point $p\in \mathbb{R}^3$ \emph{sees a positive portion} of a line segment $s\subset \mathbb{R}^3$ if there exists a subsegment $s'\subset s$ of positive length such that $p$ sees every point in $s'$. In our results, $p$ may be any point in $\mathbb{R}^3$, i.e., it may be in the interior or exterior of $\mathcal{P}$, or possibly on the boundary of $\mathcal{P}$.

\begin{theorem}\label{thm:preli1}
Let $\mathcal P$ be a polyhedron in $\mathbb{R}^3$. Then every point $p\in \mathbb{R}^3$ sees at least six edges of $\mathcal P$, and this bound is tight.
\end{theorem}

The bound in \cref{thm:preli1} is attained, for example, for a tetrahedron $\mathcal{P}$, which has six edges. In this case, a point $p$ in the exterior of $\mathcal{P}$ but close to the center of a face $F$ sees positive portions of all three edges of $F$, but it sees only an endpoint of the remaining three edges. If we insist on seeing positive portions of edges, we have the following result.

\begin{theorem}\label{thm:preli2}
Let $\mathcal P$ be a polyhedron in $\mathbb{R}^3$. Every point in $\mathcal{P}$ sees positive portions of at least six distinct edges of $\mathcal{P}$; every point in the exterior of $\mathcal{P}$ sees positive portions of at least three distinct edges of $\mathcal{P}$. Both bounds are tight.
\end{theorem}

Our main result pertains to points that do not see any vertices of the polyhedron.

\begin{theorem}\label{thm:main1}
Let $\mathcal P$ be a polyhedron in $\mathbb{R}^3$.
If a point $p\in \mathbb{R}^3$ does not see any vertex of $\mathcal P$, then it sees positive portions of at least eight distinct edges of $\mathcal P$.
The bound is tight.
\end{theorem}

\paragraph{Remark.} When $\mathcal P$ is an \emph{orthogonal polyhedron}, i.e., when every facet is orthogonal to a coordinate axis, then any point in $\mathcal P$ sees at least twelve edges (while any exterior point sees at least eight edges). Indeed, take the three planes through an internal point $p$ with normal vectors parallel to the coordinate axes; each of these planes intersects at least four edges visible to $p$. The bound is tight: It is attained, for example, in a cube; it can also be attained for points that do not see any vertices, for example in Seidel's polyhedron~\cite[Chap.~10]{o1987art}.

\paragraph{Polygonal Scenes.} In our lower bound construction for Theorem~\ref{thm:main1} (in \cref{sec:main}), a point $p$ sees eight distinct edges in six pairwise disjoint facets of a polyhedron. We extend Theorem~\ref{thm:main1} to arrangements of polygons. A \emph{polygonal scene} in $\mathbb{R}^3$ is a nonempty finite collection of polygons whose relative interiors are pairwise disjoint; see~\cref{fig:2}(left) for an example. In particular, the set of facets of a polyhedron is a polygonal scene. However, we cannot use the same definition of visibility for polygonal scenes, since interior and exterior are no longer meaningful. Given a polygonal scene $\mathcal{P}$, two points $p,q\in \mathbb{R}^3$ mutually \emph{see} each other if the line segment $pq$ does not \emph{cross} any polygon in $\mathcal{P}$, i.e., there is no polygon $P\in \mathcal{P}$ such that $p$ and $q$ lie in distinct open halfplanes bounded by the plane containing $P$ and $pq$ intersects $P$.

If we interpret the facets of a polyhedron as a polygonal scene, then the edges and vertices of polygons not coplanar with the viewpoint are always ``opaque'', and so two points cannot see each other if the line segment connecting them grazes an edge of $\mathcal P$. Thus, our definition of visibility for polygonal scenes is slightly more restrictive than the definition of visibility for polyhedra (although this distinction has no effect on our theorems).

Note that if $p$ lies in a polygon $P$ of a polygonal scene $\mathcal{P}$, then $p$ sees some vertices in the plane containing $P$. Similarly, if all polygons in $\mathcal{P}$ are coplanar, then every point $p$ (on or off that plane) sees all vertices.
Thus, when we consider a point $p$ that does not see any vertex of (any polygon in) $\mathcal{P}$, we may assume that $\mathcal{P}$ contains two or more polygons, and $p$ is not contained in any polygon in $\mathcal{P}$.

\begin{theorem}\label{thm:main2}
Let $\mathcal P$ be a polygonal scene in $\mathbb{R}^3$.
If a point $p\in \mathbb{R}^3$ does not see any vertex of $\mathcal P$, then it sees positive portions of at least eight distinct edges of $\mathcal P$.
The bound is tight.
\end{theorem}

\paragraph{Organization.} \cref{thm:preli1,thm:preli2} are proved in \cref{sec:simple,sec:positive}, respectively.
In \cref{sec:SOD}, we review \emph{Spherical Occlusion Diagrams (SODs)}, a geometric structure introduced in \cite{vigliettaSOD}. By extending the theory developed in \cite{vigliettaSOD}, we prove some properties of SODs which constitute the technical crux of the proof of \cref{thm:main1,thm:main2}. \cref{sec:conclusion} concludes the paper with some open problems.

\subsection{Related Previous Work} 
To place our results in perspective, we briefly review related work on visibility problems in combinatorial and computational geometry.

\paragraph{Triangulations.} 
A triangulation~\cite{LS17,Tri10} of a polyhedron $\mathcal{P}$ in $\mathbb{R}^d$ is a subdivision of the interior of $\mathcal{P}$ into simplices whose vertices are vertices of $\mathcal{P}$.
It is well known that every polygon in the plane can be triangulated; as noted above, this implies that every point in $\mathcal{P}$ can see at least three vertices of $\mathcal{P}$.
However, in dimensions $d\geq 3$, there exist nonconvex polyhedra that cannot be triangulated, the Sch\"onehardt polyhedron being the classical example~\cite{BC16,Sch28}. In general, it is an NP-hard problem to determine whether a given polyhedron can be triangulated~\cite{RuppertS92}. For convex polyhedra, it is NP-hard to find the minimum number of simplices in a triangulation~\cite{BelowLR00,BelowLR04}.

\paragraph{Art Gallery Problems.}
Typical art gallery problems ask for the minimum number of guards that can visually cover the interior of a polygon in the plane, or polyhedron in three-space, or some other environment in the presence of opaque obstacles. A celebrated theorem by Chv\'atal~\cite{chvatal1975combinatorial} shows that every simple polygon with $n$ vertices can be covered by at most $\lfloor n/3\rfloor$ guards located at vertices (i.e., \emph{vertex guards}). If a simple polygon is \emph{orthogonal} (i.e., every edge is parallel to a coordinate axis), $\lfloor n/4\rfloor$ vertex guards suffice~\cite{kahn1983traditional}. Finding the minimum number of vertex guards for a given polygon is NP-complete~\cite{LL86}, and there is an $O(\log \log \mathrm{OPT})$-approximation algorithm for this problem. However, if the guards can be arbitrary points in the plane (i.e., \emph{point guards}), then the problem is known to be $\exists\mathbb{R}$-complete~\cite{AAM22}, and only an $O(\log \mathrm{OPT})$-approximation algorithm is available~\cite{BM17}.

As noted above, covering a 3-dimensional polyhedron by vertex guards may be infeasible. There are some initial results on the minimum number of \emph{edge guards}~\cite{viglietta2011,cano+,Iwa17,viglietta2020} and \emph{face guards}~\cite{IKM12b,IKM12a,SVW11,Vig14}, under the notion of weak visibility: An edge or a face $f$ \emph{sees} a point $p$ if some point $s\in f$ sees the point $p$. The problem formulation for edge and face guards can be further refined depending on whether (topologically) open or closed edges and faces are allowed. However, none of the current bounds is known to be tight. Minimizing the number of edge or face guards is known to be NP-hard in several variants of the problem~\cite{Iwa17,IKM12a,Vig14}. The current best bounds~\cite{cano+} for the minimum number of edge guards in a polyhedron $\mathcal{P}$ in $\mathbb{R}^3$ distinguish between points visible and invisible to vertices: Theorem~\ref{thm:main1} addresses the latter scenario, albeit it does not improve on the current bounds for the edge guard problem.

\paragraph{Hidden Surface Removal.} The computational counterpart of our results is a classical problem in computer vision~\cite{Ber95}: Given an arrangement of polyhedral or polygonal objects (a \emph{polyhedral} or \emph{polygonal scene}) and a light source $s\in \mathbb{R}^3$, compute the parts of the objects that are visible to $s$, i.e., the \emph{visibility map} of $s$. For a polyhedral scene in $\mathbb{R}^3$ with a total of $n$ vertices, edges, and faces, the complexity of the visibility map is $O(n^2)$ and this bound is tight; it can be computed in $O(n^2)$ time~\cite{KMS+02,McK87}. However, one can find the vertices, edges, and faces that are (at least partially) visible to $s$ in $O(n\log n)$ time~\cite{GMV99} in an arrangement of axis-aligned rectangles. In particular, the \emph{visibility counting} problem asks for the number of faces (respectively, edges or vertices) visible to $s$~\cite{bygi2015weak,GM10}. There are many results on data structures that preprocess a polyhedral scene to support a fast computation of the visibility map for a query point $s\in \mathbb{R}^3$; there are also bounds on the number of combinatorially different visibility maps for a given scene~\cite{BHO+97}; refer to the surveys~\cite{OR17,urrutia2000art} for further references. Analogous problems were also considered when the light source $s$ is one or more line segments or triangles~\cite{DDE+09,MNK08}, corresponding to edge guards or face guards.

\section{Proof of \cref{thm:preli1}}\label{sec:simple}

In the \emph{brush} polyhedron depicted in \cref{fig:1:1}, any interior point close to the tip of a tetrahedral ``spike'' sees exactly six edges. In this section, we prove that any point $p\in \mathbb{R}^3$ sees at least six edges of \emph{any} polyhedron. We remark that this also holds if $p$ is in the exterior of the polyhedron.

 \begin{figure}[ht!]
	\centering
	\begin{subfigure}[t]{0.4\linewidth}
		\centering
		\includegraphics[height=3.2cm] {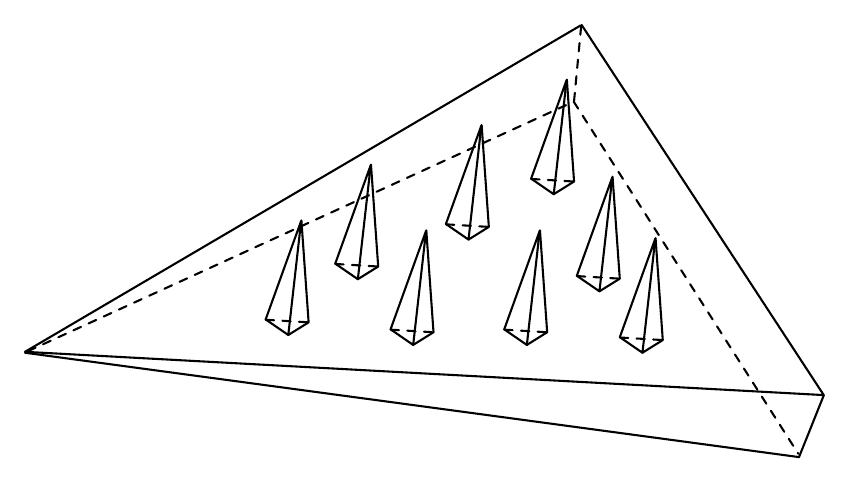}
		\caption{A brush}
		\label{fig:1:1}
	\end{subfigure}
	~~~~~
	\begin{subfigure}[t]{0.5\linewidth}
		\centering
		\includegraphics[height=3.5cm] {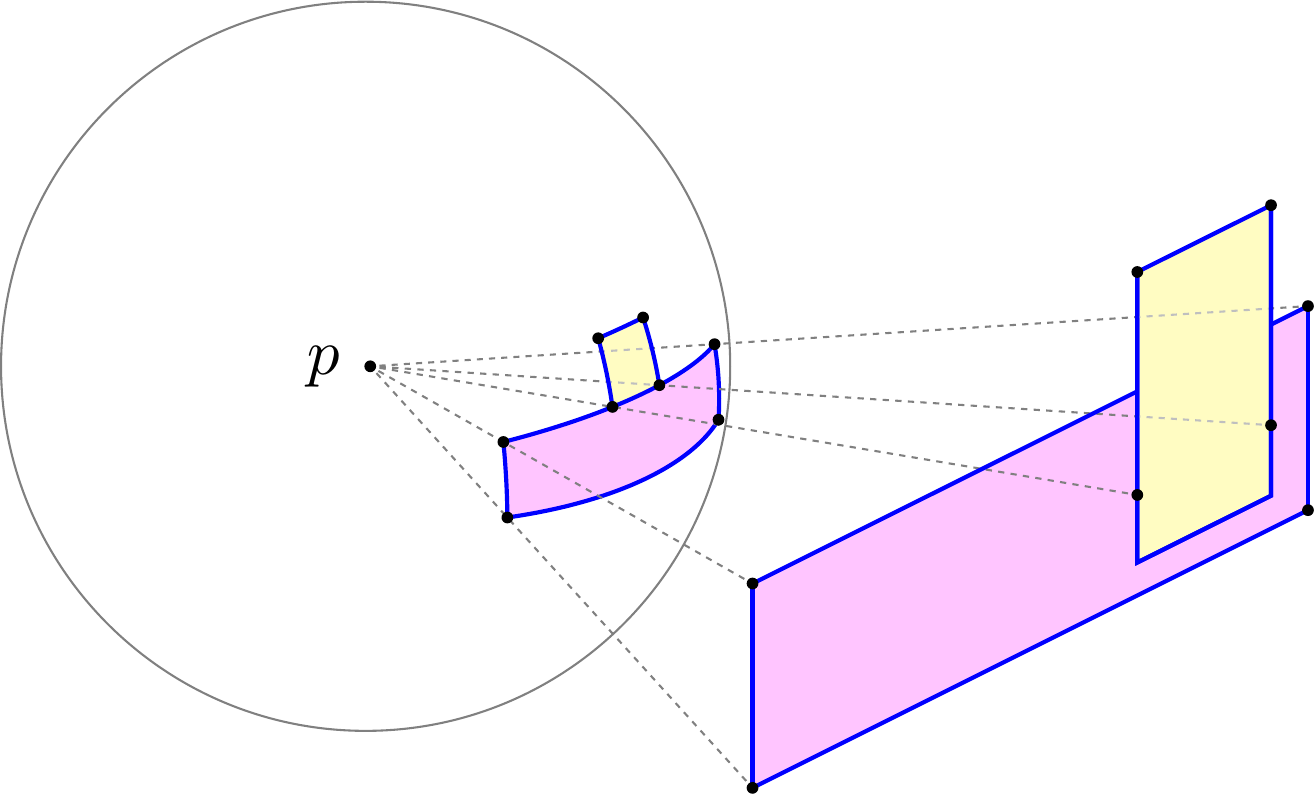}
		\caption{Orthographic projection of polygons onto a sphere}
		\label{fig:1:2}
	\end{subfigure}
	\caption{Visibility in $\mathbb{R}^3$}
	\label{fig:1}
\end{figure}

Let us fix a polyhedron $\mathcal P$ and a point $p\in \mathbb{R}^3$. For each edge $e$ of $\mathcal P$,
we consider the set of points in $e$ that are not occluded by other edges: Let $E_{e,p}$ be the set of points $x\in e$ such that $x$ is visible to $p$ and the relative interior of $xp$ is disjoint from every closed edge of $\mathcal{P}$. Note that $E_{e,p}$ is the union of finitely many pairwise disjoint line segments along $e$.
Let $\mathcal V_{\mathcal P,p}$ be the set of all segments in $e$ corresponding to $E_{e,p}$ over all edges $e$ of $\mathcal P$. Some segments in $\mathcal V_{\mathcal P,p}$ may be ``degenerate'', i.e., they may be single points. However, if $\{x\}\subset E_{e,p}$ is a degenerate segment, then $x$ is contained in at least two edges, consequently it is a vertex of $e$. In particular, every degenerate segment in $\mathcal V_{\mathcal P,p}$ coincides with a vertex of $\mathcal{P}$.

For example, if $\mathcal P$ is a regular tetrahedron and $p$ is a point in the exterior of $\mathcal P$ near the center of a face $F$, then $\mathcal V_{\mathcal P,p}$ consists of the three edges of $F$ plus three degenerate segments coinciding with the three vertices of $F$ (each of these corresponds to an edge of $\mathcal P$ that is hidden behind $F$, only one of its endpoints being visible to $p$).

\begin{lemma}\label{l:6edge1}
$\mathcal V_{\mathcal P,p}$ consists of at least six segments.
\end{lemma}
\begin{proof}
Let $S$ be the unit sphere centered at $p$, and let us orthographically project all the segments in $\mathcal V_{\mathcal P,p}$ onto $S$; \cref{fig:1:2} shows an example of an orthographic projection. The projection of each segment $s\in \mathcal V_{\mathcal P,p}$ is an arc $a_s$ of a great circle of $S$; let $\mathcal A$ be the collection of these arcs. Note that $\mathcal A$ is a noncrossing arrangement of arcs, some of which may be degenerate, i.e., single points. Each nondegenerate arc of $\mathcal A$ has two endpoints, each of which lies on at least one other arc of $\mathcal A$. Each arc in $\mathcal A$ is the projection of a line segment, hence it is shorter than a great semicircle; every face of the arrangement is the projection of a polygon in $\mathbb{R}^3$, hence it is contained in a hemisphere. Therefore, any two arcs in $\mathcal{A}$ intersect in at most one point; in other words, $\mathcal{A}$ does not contain any \emph{lens}, intended as a pair of internally disjoint curves that share more than one point. \cref{fig:6edge:1} depicts (a plane drawing of) an arrangement $\mathcal A$, where the dashed line represents a degenerate arc, which originates from an edge of $\mathcal P$ that is hidden behind a face, except for an endpoint.

\begin{figure}[ht!]
	\centering
	\begin{subfigure}[t]{0.45\linewidth}
		\centering
		\includegraphics[scale=0.75] {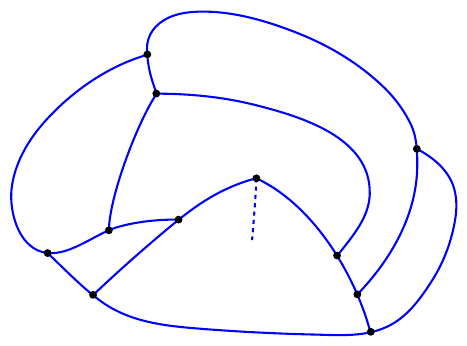}
		\caption{A plane drawing of $\mathcal A$, including a degenerate arc represented as a dashed line}
		\label{fig:6edge:1}
	\end{subfigure}
	~~
	\begin{subfigure}[t]{0.45\linewidth}
		\centering
		\includegraphics[scale=0.75] {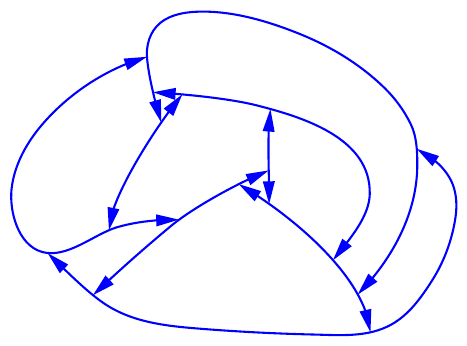}
		\caption{The arrangement $\mathcal A'$}
		\label{fig:6edge:2}
	\end{subfigure}\\
	\vspace{0.5cm}
	\begin{subfigure}[t]{0.45\linewidth}
		\centering
		\includegraphics[scale=0.75] {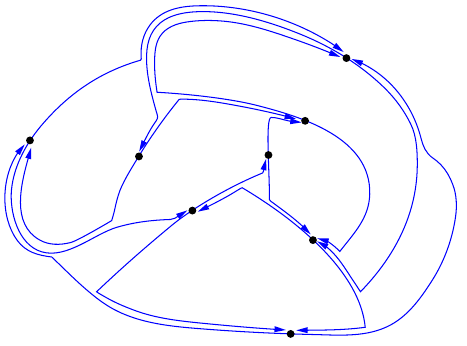}
		\caption{The contact graph $\mathcal G$ of $\mathcal A'$}
		\label{fig:6edge:3}
	\end{subfigure}
	~~
	\begin{subfigure}[t]{0.45\linewidth}
		\centering
		\includegraphics[scale=0.75] {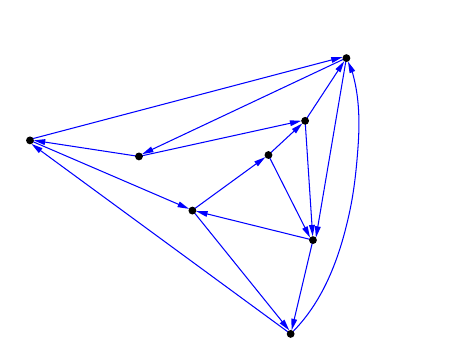}
		\caption{Another plane drawing of $\mathcal G$; note that every vertex has out-degree $2$}
		\label{fig:6edge:4}
	\end{subfigure}
	\caption{Proving that $\mathcal V_{\mathcal P,p}$ consists of at least six segments}
	\label{fig:6edge}
\end{figure}

Let us transform $\mathcal A$ as follows. First, we successively extend each degenerate arc in $\mathcal A$ (which, we recall, coincides with the projection of a vertex of $\mathcal P$ onto $S$) along a great circle until it hits another arc of $\mathcal A$. Since the extended arcs each lie in some face of the initial arrangement, they are each shorter than a semicircle, which in turn implies that no lens is created. After this step, $\mathcal{A}$ is an arrangement of pairwise noncrossing nondegenerate arcs without lenses.

We further modify the arrangement $\mathcal{A}$ to an arrangement $\mathcal{A}'$ as follows. (The arcs in $\mathcal{A}'$ are not necessarily arcs of great circles). For each endpoint $v$ of an arc in $\mathcal A$ that does not lie in the relative interior of any arc of $\mathcal A$, we do the following. Let $a_1, a_2, \dots, a_k$ be the arcs incident to $v$, taken in clockwise or counterclockwise order. By construction, $v$ is either the projection of a vertex of $\mathcal P$, or the endpoint of an arc that has been extended. In both cases, we have $k\geq 3$. Now deform $a_1$, $a_2$, \dots, $a_k$ in a small neighborhood of $v$ so that the endpoint of $a_i$ is re-routed to the interior of $a_{i+1}$, for all $1\leq i\leq k$, where indices are taken modulo $k$. This operation can be done without creating crossings or lenses (since $k\geq 3$).

The resulting arrangement $\mathcal A'$ is illustrated in \cref{fig:6edge:2}. Note that $\mathcal A'$ is still a noncrossing arrangement of Jordan arcs without lenses. In particular, (i)~each endpoint of each arc is in the interior of exactly one other arc, (ii)~no arc has both endpoints in the interior of the same arc, and (iii)~if an arc $a$ has an endpoint in the interior of an arc $b$, then $b$ does not have an endpoint in the interior of $a$.

Next we define the \emph{contact graph} $\mathcal G$ of $\mathcal A'$ as follows. Let $\mathcal{G}$ be a directed graph, where the vertices correspond to the arcs in $\mathcal A'$, and there is a directed edge $(a,b)$ in $\mathcal G$ if and only if the arc $a\in \mathcal A'$ has an endpoint in the interior of the arc $b\in \mathcal A'$ (see \cref{fig:6edge:3}). Due to the properties of $\mathcal A'$, the contact graph $\mathcal G$ is a simple planar directed graph where each vertex has out-degree exactly $2$. Also, the number of vertices of $\mathcal G$ is equal to the number of segments in $\mathcal V_{\mathcal P,p}$. In particular, $\mathcal G$ is nonempty, because $p$ sees at least one nondegenerate segment of an edge of $\mathcal P$.

We will now prove that $\mathcal G$ has at least six vertices, which implies that $\mathcal V_{\mathcal P,p}$ consists of at least six segments, as well. Note that if $\mathcal G$ has $n\geq 1$ vertices, it has exactly $2n$ edges. Since a simple graph on $n$ vertices can have at most $\binom{n}{2}$ edges, then $2n\leq n(n-1)/2$ implies $n\geq 5$. Moreover, the only $5$-vertex simple graph with $10$ edges is the complete graph, which is not planar. We conclude that $n\geq 6$, as claimed.
\end{proof}

\begin{lemma}\label{l:6edge2}
If there is an edge $e$ of $\mathcal P$ such that $E_{e,p}$ is disconnected, then $p$ sees at least six distinct edges of $\mathcal P$.
Furthermore, if two or more connected components of $E_{e,p}$ have positive length, then $p$ sees positive portions of at least six distinct edges of $\mathcal P$.
\end{lemma}
\begin{proof}
First assume that $E_{e,p}$ consists of two or more segments of positive length, and let $s_1$ and $s_2$ be two such segments in $E_{e,p}$. Observe that $p$ is not collinear with $e$, or $E_{e,p}$ would be connected. Thus, $p$ and $e$ determine a unique plane $\alpha$. Consider the cross section $P$ of $\mathcal P$ corresponding to $\alpha$, illustrated in \cref{fig:aligned}.

\begin{figure}[h]
\centering
\includegraphics[scale=0.8]{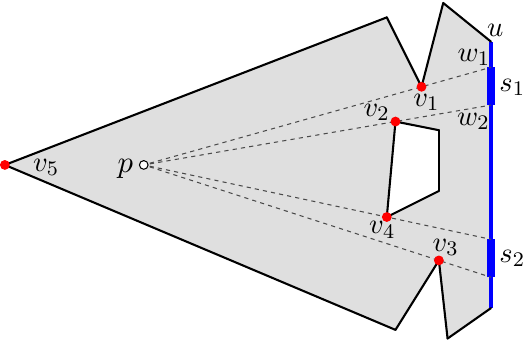}
\caption{If $p$ sees two sub-segments $s_1$ and $s_2$ of the same edge $e$ of $\mathcal P$, it sees at least six edges.}
\label{fig:aligned}
\end{figure}

Let $u$ be the endpoint of $e$ such that $u$ is closer to $s_1$ than to $s_2$, and let $w_1$ and $w_2$ be the endpoints of $s_1$, with $w_1$ closer to $u$. Clearly, there must be a vertex $v_2$ of $P$ on the segment $pw_2$. Also, there must be a vertex $v_1$ of $P$ on the segment $pw_1$, possibly with $v_1=u=w_1$. In both cases, there is an edge of $\mathcal P$ other than $e$ that intersects $\alpha$ in $v_i$, for $i\in \{1,2\}$, which is visible to $p$ in the plane $\alpha$.
Furthermore, a small neighborhood of $v_1$ and $v_2$, respectively, $p$ sees a positive portion of two edges of $\mathcal{P}$. In the special case that $v_1=u=w_1$, the point $p$ sees the interior of some face $F$ incident to $u$ and $e$, and $p$ sees a positive portion of another edge of $F$ incident to $v_1$.
Similarly, $s_2$ determines two additional vertices of $P$, say $v_3$ and $v_4$, which are visible to $p$ and contained in some edges of $\mathcal{P}$.
%

Now consider the line $\ell$ through $p$ and parallel to $e$, and let $\ell^+\subset \alpha$ denote the closed halfplane bounded by $\ell$ that does not contain $e$. Clearly, $\ell^+$ does not contain vertices $v_1,\ldots ,v_4$, either.
%
%
Note, however, that $P$ must have at least one vertex in $\ell^+$ visible to $p$. Indeed, let $S$ be a circle centered at $p$, and project orthographically all the segments on the boundary of $P$ that are visible to $p$ onto $S$. Since each segment projects to a circular arc shorter than a semicircle, then $S\cap \ell^+$ contains the endpoint of an arc, which is the projection of a vertex $v_5$ of $P$ in $\ell^+$.

We have determined five vertices, $v_1$, $v_2$, $v_3$, $v_4$, and $v_5$, that correspond to distinct edges of $\mathcal P$ of which $p$ sees at least a positive portion. Since $e$ is distinct from these five edges, we conclude that $p$ sees a positive portion of at least six edges of $\mathcal P$.

Now suppose that $s_1$ or $s_2$ has zero length, i.e., it is an endpoint of the edge $e$. If $u=w_1=w_2$, then $u$ is the endpoint of at least three edges of $\mathcal{P}$: edge $e$ and two additional edges; we may assume that $v_1$ and $v_2$ correspond to two edges incident to $u$ other than $e$. The same argument holds for $s_2$. We conclude that $p$ sees at least six edges of $\mathcal{P}$ (albeit, possibly only one endpoint of some of these edges).
\end{proof}

\cref{thm:preli1} now follows immediately. If some $E_{e,p}$ has two or more components, then \cref{l:6edge2} implies that $p$ sees at least six distinct edges of $\mathcal P$. Otherwise, $p$ sees at least as many distinct edges of $\mathcal P$ as there are segments in $\mathcal V_{\mathcal P,p}$; thus, by \cref{l:6edge1}, $p$ sees at least six edges.

\section{Proof of \cref{thm:preli2}}\label{sec:positive}

Every point in a tetrahedron $\mathcal{P}$ sees positive portions of all six edges; a point in the exterior of $\mathcal{P}$ but close to the center of a face $F$ of $\mathcal{P}$ sees positive fractions of only three edges. This shows that both bounds in \cref{thm:preli2} are tight.

For the lower bounds in \cref{thm:preli2}, let $\mathcal{P}$ be a polyhedron and $p$ a point in $\mathbb{R}^3$.
If $p \notin \mathcal{P}$, then consider any point $q$ in the interior of $\mathcal{P}$ in general position with respect to $\mathcal{P}$. Then $p$ sees the interior of a face of $\mathcal{P}$ along the ray $\overrightarrow{pq}$, and it sees either another face or infinity in the opposite direction. This means that the visibility map of $p$ (in an orthographic projection to a sphere $S$ centered at $p$) has at least two faces, hence it consists of arcs of at least three distinct great circles. These arcs correspond to three distinct edges of $\mathcal{P}$, some positive portions of which are visible to $p$.

If $p$ is coplanar with a facet $F$ of $\mathcal{P}$, then $p$ sees at least three vertices and positive portions of at least three edges of $F$. Furthermore, each visible vertex $v$ of $F$ is incident to an edge $e_v$ that is not coplanar with $F$ and whose initial portion is visible to $p$. There are at least three distinct such edges, as $v$ sees at least three vertices of $F$. Altogether, $v$ sees positive portions of at least six edges of $\mathcal{P}$.

In the remainder of this section, we may assume that $p$ lies in the interior of $\mathcal{P}$.
We follow the proof of \cref{thm:preli1} (\cref{sec:simple}), but without degenerate segments.
Recall that for each edge $e$ of $\mathcal P$, we denote by $E_{e,p}$ the set of points $x\in e$ such that $x$ is visible to $p$ and the relative interior of $xp$ is disjoint from every closed edge of $\mathcal{P}$,
where $E_{e,p}$ is the union of finitely many pairwise disjoint line segments along $e$.  Let $\mathcal W_{\mathcal P,p}$ be the set of all segments \emph{of positive length} in $e$ corresponding to $E_{e,p}$ over all edges $e$ of $\mathcal P$. Finally, let $\mathcal A$ be the collection of orthographic projections of the segments in $\mathcal W_{\mathcal P,p}$ onto a sphere $S$ centered at $p$. As noted in \cref{sec:simple}, $\mathcal{A}$ is an arrangement of pairwise noncrossing arcs of great circles. Each arc is strictly shorter than a semicircle, and so they cannot form lenses. Furthermore, as $p$ lies in the interior of $\mathcal{P}$, every face of the arrangement $\mathcal{A}$ is strictly contained in a hemisphere of $S$.

It is sufficient to show that $\mathcal W_{\mathcal P,p}$ contains at least six segments, and then Lemma~\ref{l:6edge2} completes the proof. Let $n=|\mathcal W_{\mathcal P,p}|=|\mathcal{A}|$. We distinguish between two cases.

\paragraph{Case~1: There exists a facet $F$ of $\mathcal{P}$ that contains three or more segments of $\mathcal{W}_{\mathcal{P},p}$.} Let $\alpha$ be a plane parallel to $F$ and containing $p$. The cross section of $\mathcal{P}$ corresponding to $\alpha$ consists of one or more polygons. If we triangulate the polygons arbitrarily, then $p$ lies in one of the triangles, and so $p$ sees at least three vertices of $P$. These vertices correspond to three distinct edges of $\mathcal{P}$, in which $p$ sees a positive portion, hence they each contain a segment in $\mathcal{W}_{\mathcal{P},p}$. Since none of these segments is in the plane $\alpha$, then $\mathcal{W}_{\mathcal{P},p}$ contains at least six segments, as required.

\paragraph{Case~2: Every facet of $\mathcal{P}$ contains at most two segments of $\mathcal{W}_{\mathcal{P},p}$.}
Let $Q$ be the set of points $q\in S$ such that $q$ is the endpoint of precisely two arcs in $\mathcal{A}$ but does not lie in the relative interior of any arc in $\mathcal{A}$. Each point $q\in Q$ is the orthographic projection of a vertex of $\mathcal{P}$ visible to $p$, and the two incident arcs are orthographic projections of two edges of some facet $F_v$ incident to $v$. Note that the two arcs incident to $v$ cannot lie on the same great circle, otherwise $p$ would be coplanar with the facet $F_v$, and $p$ would see a third edge incident to $v$; hence, $q$ would be the endpoint of three or more arcs in $\mathcal{A}$.

Further note that no two points in $Q$ are connected by an arc in $\mathcal{A}$, otherwise the three arcs incident to these points would be the projections of three segments in $\mathcal{W}_{\mathcal{P},p}$ on the boundary of the same facet. Since each point $q\in Q$ is incident to two arcs, and each arc is incident to at most one point in $Q$, we conclude  that $2\, |Q|\leq n$.

We modify $\mathcal{A}$ as follows. For each point $q\in Q$ consider the two incident arcs, say $a_1,a_2\in \mathcal{A}$. Extend one of them, say $a_1$ along a great circle beyond $q$ until it hits another arc.
Note that $a_1$ was originally shorter than a great semicircle; also, both $a_1$ and its extension lie in a same face of $S$, which is contained in a hemisphere. Thus, the extended arc is still shorter than a great semicircle. We obtain a pairwise noncrossing arrangement $\mathcal{A}'$ of arcs of great circles, each of which is shorter than a semicircle.
This implies that the $\mathcal{A}'$ does not contain lenses.
%

We further modify the arrangement $\mathcal{A}'$ as in \cref{sec:simple}: For every point $s\in S$ that is the endpoint of three or more arcs in $\mathcal A'$ but does not lie in the relative interior of any arc of $\mathcal A'$, we perturb the arrangement in a neighborhood of $s$ so that the endpoint of each arc incident to $s$ is re-routed to the interior of the next arc (refer to \cref{fig:6edge:2}). Denote by $\mathcal{A}''$ the resulting arrangement of noncrossing arcs, and let $\mathcal{G}$ be its contact graph introduced in \cref{sec:simple}.

Recall that $\mathcal{G}$ is a directed planar graph where every node has outdegree two; thus, it has $n=|\mathcal{A}|$ vertices and $2n$ edges. Since $\mathcal{A}'$ does not contain lenses, $\mathcal{G}$ is a simple graph, and $n\geq 6$ as in the proof of \cref{l:6edge1}.

\section{Spherical Occlusion Diagrams}\label{sec:SOD}


To prove our next results, we study \emph{Spherical Occlusion Diagrams} (SOD), introduced in~\cite{vigliettaSOD}. SODs arise from the visibility map of a point $p$ that does not see any vertex of a polygonal scene $\mathcal P$. Specifically, a \emph{visibility map} generated by $\mathcal P$ with viewpoint $p$ is the set of arcs of great circles obtained by orthographically projecting all edge sub-segments of $\mathcal P$ that are visible to $p$ onto the unit sphere centered at $p$. \cref{fig:1:2} shows the orthographic projection of two rectangles onto a sphere, and \cref{fig:2} shows the visibility map generated by the polyhedron in the middle picture, as well as the polygonal arrangement in the left picture, both with viewpoint $p$ at the center of the arrangement.

\begin{figure}
\centering
\includegraphics[width=\linewidth]{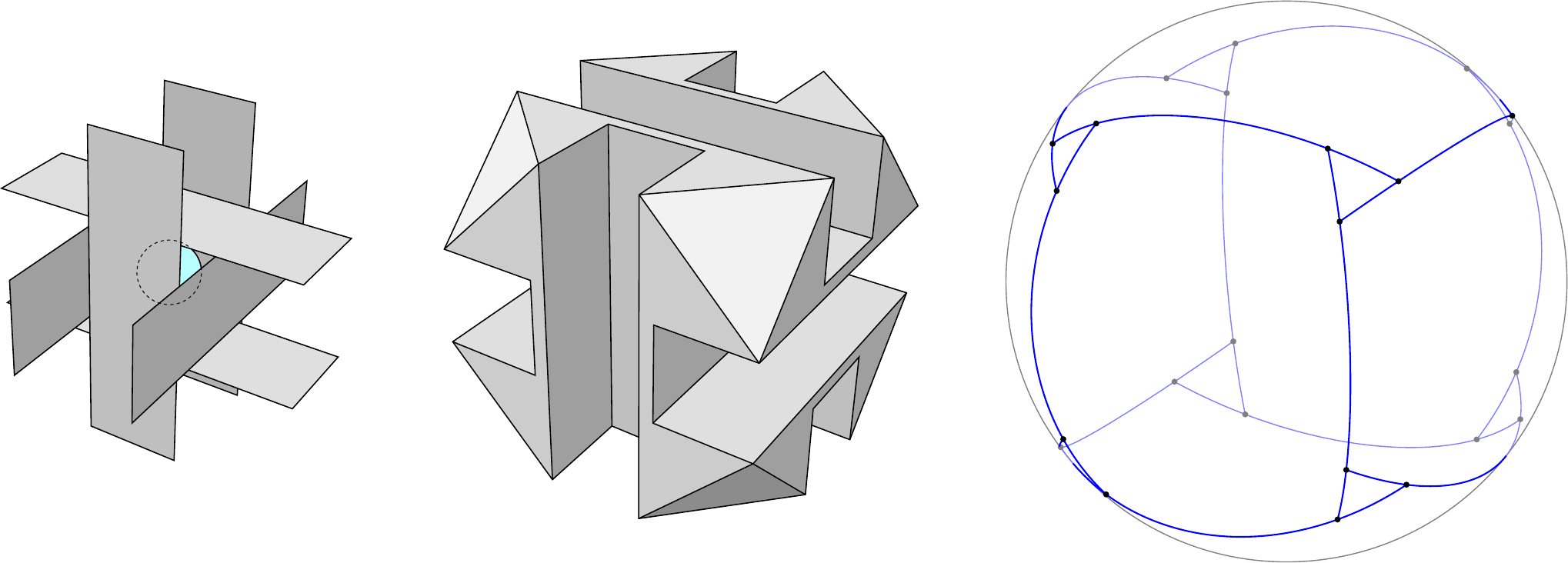}
\caption{The SOD in the right picture is generated by the arrangement of internally disjoint rectangles in the left picture, where the viewpoint is at the center of the arrangement. The polyhedron in the middle picture generates the same SOD.}
\label{fig:2}
\end{figure}

A SOD arises from an axiomatic formulation of basic properties of a visibility map from a viewpoint $p$. Formally, a SOD is defined as a finite non-empty set of arcs of great circles on a sphere satisfying the following properties, which are called \emph{diagram axioms} in~\cite{vigliettaSOD}:
\begin{itemize}
\item[(A1)] Each arc is shorter than a great semicircle, and any two arcs are internally disjoint.
\item[(A2)] Each arc feeds into another arc at each endpoint.
\item[(A3)] All arcs that feed into the same arc reach it from the same side.
\end{itemize}
An arc $a$ is said to \emph{feed into} another arc $a'$ (equivalently, $a$ \emph{hits} $a'$ and $a'$ \emph{blocks} $a$) if an endpoint of $a$ is in the relative interior of $a'$.

It is shown in~\cite{vigliettaSOD} that the visibility map from any arbitrarily located viewpoint $p$ that does not see any vertex of $\mathcal P$ is indeed a SOD. (Incidentally, it is known that not every SOD is a visibility map~\cite{KT22}.) We summarize some other basic results from~\cite{vigliettaSOD}:
\begin{theorem}\label{thm:SOD}
The following statements hold for every SOD $\mathcal S$.
\begin{itemize}
\item[(1)] No two arcs of $\mathcal S$ feed into each other.
\item[(2)] The union of all the arcs of $\mathcal S$ is a connected set.
\item[(3)] If $\mathcal S$ has $n$ arcs, it partitions the sphere into $n + 2$ spherically convex regions.
\item[(4)] The relative interior of any great semicircle intersects some arc of $\mathcal S$.\qed
\end{itemize}
\end{theorem}

A common structure in a SOD is the \emph{swirl}, which is defined as a cycle of arcs such that each arc feeds into the next, going always clockwise or always counterclockwise. In the SOD in \cref{fig:2}, there are four clockwise swirls and four counterclockwise swirls.
The \emph{swirl graph} of a SOD is the undirected multigraph on the
set of swirls such that, for each arc shared by two swirls, there is an edge in
the swirl graph. The following theorem is proved in~\cite{vigliettaSOD}:

\begin{theorem}\label{thm:swirls}
The swirl graph of a SOD is a simple planar bipartite graph with nonempty partite sets.\qed
\end{theorem}

We will now extend \cref{thm:swirls} by proving that the swirl graph contains at least four swirls. The spherically convex region enclosed by a swirl is called the \emph{eye} of that swirl. For example, \cref{fig:4} shows the eye of a swirl in yellow.

\begin{lemma}\label{l:lemma}
Given a SOD $\mathcal S$, every hemisphere contains the eye of at least one swirl of $\mathcal S$.
\end{lemma}
\begin{proof}
Let $H$ be the interior of a hemisphere; by \cref{thm:SOD}(4), there is an arc $a_0\in\mathcal{S}$ that intersects $H$. Let us construct a ``walk'' within $H$ that traverses some arcs of $\mathcal S$ as follows (refer to \cref{fig:5}). The walk starts from a point of $a_0 \cap \partial H$. Note that at least one endpoint $p_0$ of $a_0$ is in $H$, otherwise $a_0$ would not be shorter than a great semicircle, contradicting axiom~A1. The walk follows $a_0$ to $p_0$, and then continues in this fashion: upon reaching an endpoint $p_i\in a_i\cap H$ of an arc $a_i$, proceed into the arc $a_{i+1}$ which blocks $a_i$ at $p_i$. Pick any endpoint $p_{i+1}$ of $a_{i+1}$ contained in $H$ (such an endpoint exists due to axiom~A1) and follow $a_{i+1}$ to $p_{i+1}$, and so on.

\begin{figure}[h]
\centering
\includegraphics[height=6cm]{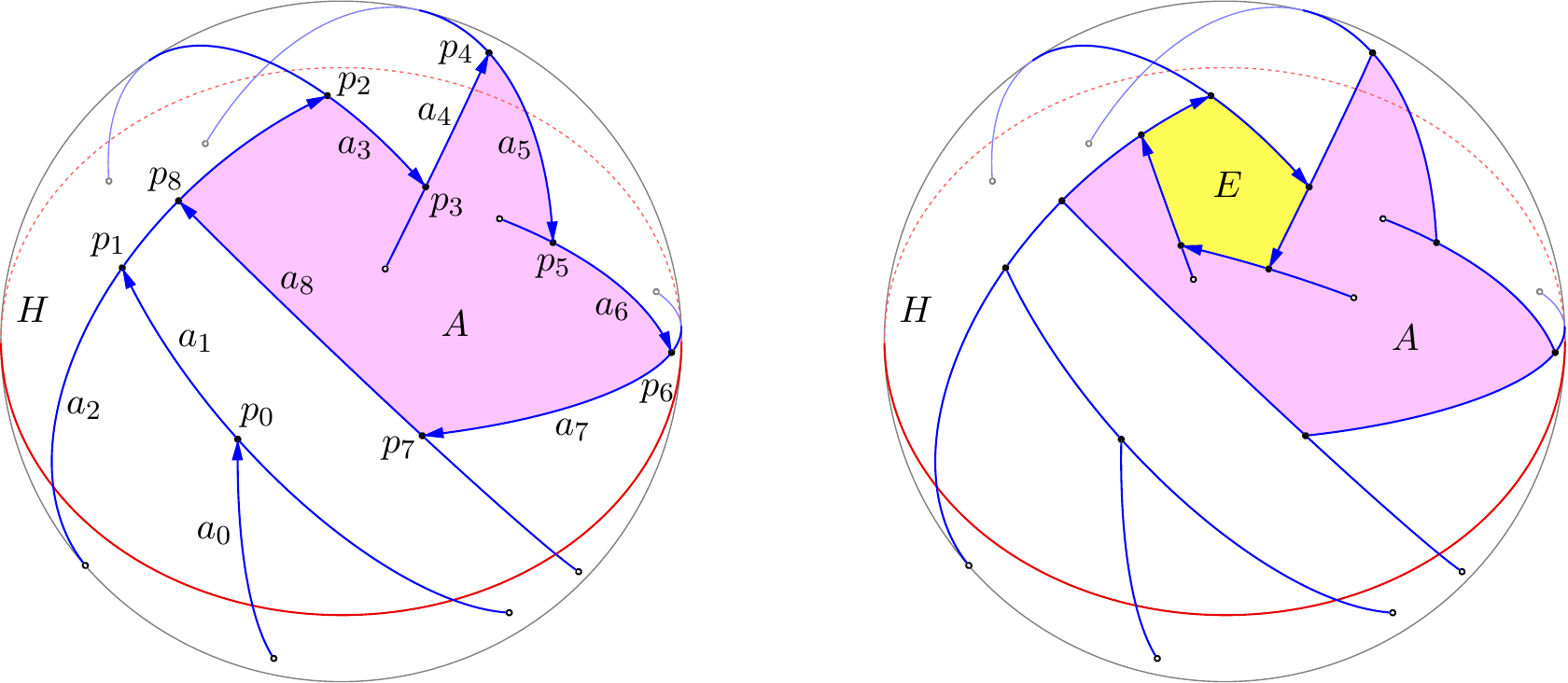}
\caption{An example of the walk within a hemisphere $H$ constructed in \cref{l:lemma}. The walk eventually encloses an area $A$ (left picture). If $A$ is on the right-hand side of the walk (i.e., if the walk travels around $A$ in the clockwise direction), then following the SOD starting from the boundary of $A$ and always turning right upon reaching the endpoint of the current arc eventually traces out the eye $E$ of a swirl which is entirely contained in $H$ (right picture).}
\label{fig:5}
\end{figure}

Since $\mathcal S$ has finitely many arcs, the walk eventually reaches a point that it has already visited. As soon as this happens, the walk has enclosed a region $A\subseteq H$. Assume, without loss of generality, that the walk travels around the boundary of $A$ in the clockwise direction, as in the left picture of \cref{fig:5} (if it is in the other direction, a symmetric argument applies).

Let us construct a second walk as follows. Starting from any point on the boundary of $A$, say on arc $a_i$, follow the first walk until an endpoint of $a_i$ is reached. Then turn right into the next arc of $\mathcal S$, follow it to its endpoint, and so on. Since the second walk always turns right, it eventually traces out the eye $E$ of a clockwise swirl, as shown in the right picture of \cref{fig:5}. Moreover, it is easy to see that the second walk is bound to remain within $A$, and therefore $E\subseteq A\subseteq H$.
\end{proof}

\begin{figure}[h]
\centering
\includegraphics[height=6cm]{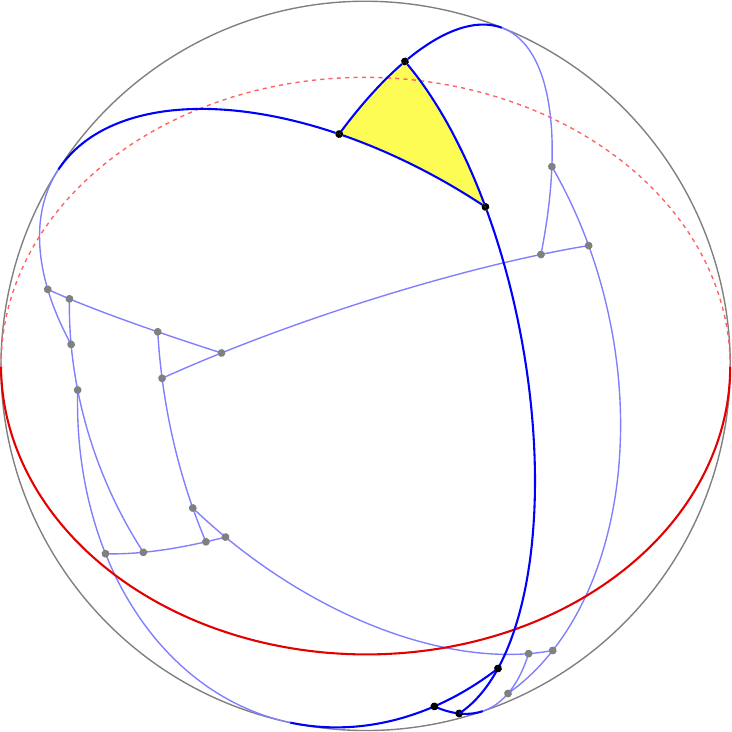}
\caption{A SOD where a hemisphere contains the eye of exactly one swirl (in yellow).}
\label{fig:4}
\end{figure}

Note that there are cases where a hemisphere contains the eye of exactly one swirl of a SOD, as shown in \cref{fig:4}.

\begin{theorem}\label{thm:swirls2}
Every SOD has at least four swirls.
\end{theorem}
\begin{proof}
By \cref{thm:swirls}, there are at least one clockwise swirl $\mathcal W_1$ and at least one counterclockwise swirl $\mathcal W_2$. Consider a great circle $G$ that intersects the interiors of the eyes of both $\mathcal W_1$ and $\mathcal W_2$. Then $G$ determines two hemispheres $H_1$ and $H_2$, none of which entirely contains the eye of $\mathcal W_1$ or $\mathcal W_2$. Hence, by \cref{l:lemma}, there must be a third swirl $\mathcal W_3$ whose eye is entirely contained in $H_1$, and a fourth swirl $\mathcal W_4$ whose eye is entirely contained in $H_2$.
\end{proof}

As a consequence we have the following corollary, which yields a weaker version of our main result, \cref{thm:main1}:

\begin{corollary}\label{cor:8arcs}
Every SOD has at least eight arcs.
\end{corollary}
\begin{proof}
Let $\mathcal S$ be a SOD; by \cref{thm:swirls2}, $\mathcal S$ has at least four swirls, each of which consists of at least three arcs. Consider four arbitrary swirls, and all arcs in $\mathcal{S}$ that are incident to at least one of these swirls. Recall that the swirl graph is simple and planar due to \cref{thm:swirls}. A simple planar bipartite graph on four vertices has at most four edges, hence at most four of these arcs are incident to two of these swirls. Thus, these four swirls involve at least $4\cdot 3- 4=8$ distinct arcs. We conclude that $\mathcal S$ has at least eight arcs.
\end{proof}
\cref{fig:2} shows an example of a SOD with exactly eight arcs and four swirls, proving that \cref{thm:swirls2,cor:8arcs} are tight.

\begin{figure}[h]
\centering
\includegraphics[height=6cm]{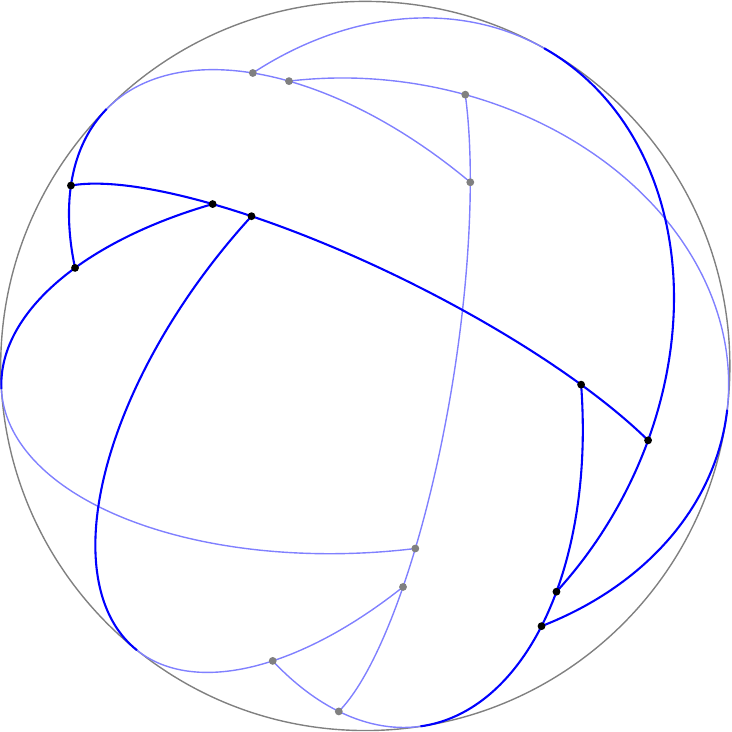}
\caption{A SOD with eight arcs and four swirls}
\label{fig:3}
\end{figure}

\section{Proof of \cref{thm:main1,thm:main2}}\label{sec:main}

We will now prove both \cref{thm:main1} and \cref{thm:main2}.

\paragraph{Upper bound.} Let $\mathcal P$ be a polygonal scene (possibly the facets of a polyhedron). Let $p$ be any point that sees no vertices of $\mathcal P$, and let $\mathcal S$ be the SOD generated by $\mathcal P$ with viewpoint $p$. We know from \cref{cor:8arcs} that $\mathcal S$ consists of at least eight arcs. However, this is insufficient to conclude that $p$ sees at least eight distinct edges of $\mathcal P$, because some arcs of $\mathcal S$ may be projections of sub-segments of the same edge of $\mathcal P$.

We will first give some definitions. Recall that a \emph{great semicircle} is an arc consisting of half of a great circle on the unit sphere. A \emph{semicircle cover} $\mathcal C$ of a SOD $\mathcal S$ is a set of great semicircles such that each arc of $\mathcal S$ is contained in the relative interior of a great semicircle in $\mathcal C$.

\begin{lemma}\label{l:semi}
Every semicircle cover of a SOD consists of at least eight great semicircles.
\end{lemma}
\begin{proof}
We say that a SOD is \emph{good} if all of its semicircle covers consist of at least eight great semicircles. We will prove that all SODs are good.

Consider the \emph{arc-addition} relation $\mathcal S_1\prec \mathcal S_2$ between SODs, meaning that the SOD $\mathcal S_2$ can be obtained by adding a single arc to the SOD $\mathcal S_1$. Since SODs have finitely many arcs, the arc-addition relation is well-founded (i.e., it has no infinite decreasing chains); also, if $\mathcal S_1\prec \mathcal S_2$ and $\mathcal S_1$ is a good SOD, then so is $\mathcal S_2$. Thus, we will prove that all SODs are good by well-founded induction with respect to the arc-addition relation.

Let $\mathcal S$ be a SOD. If there is an arc $a\in\mathcal S$ that does not block any arc of $\mathcal S$, then removing $a$ yields another SOD $\mathcal S'$, which is good by the inductive hypothesis. Since $\mathcal S'\prec \mathcal S$, we conclude that $\mathcal S$ is good, as well.

So, assume that all arcs of $\mathcal S$ block other arcs of $\mathcal S$. If there is no great semicircle whose relative interior contains two arcs of $\mathcal S$, then $\mathcal S$ is good, due to \cref{cor:8arcs}. We may therefore assume that there are two arcs $a_1,a_2\in \mathcal S$ that lie in the relative interior of a same great semicircle $C$, as shown in \cref{fig:6}.

\begin{figure}
\centering
\includegraphics[height=6cm]{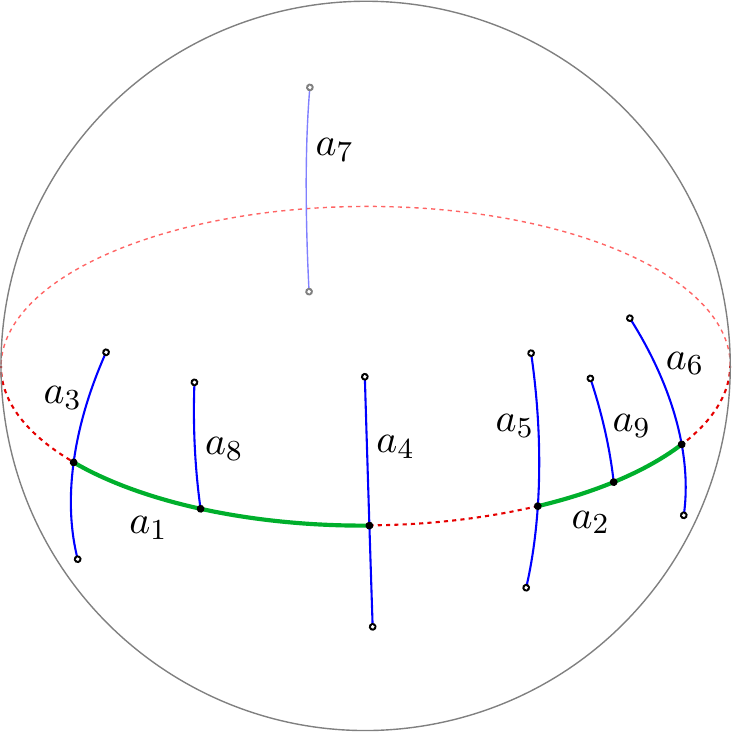}
\caption{If a great circle contains two arcs $a_1$ and $a_2$ that are projections of sub-segments of a same edge of an arrangement of polygons, there must be seven additional arcs touching the same great circle.}
\label{fig:6}
\end{figure}

The arcs $a_1$ and $a_2$ must hit four distinct arcs $a_3,a_4,a_5,a_6\in\mathcal S$. Let $G$ be the great circle containing $C$. Due to \cref{thm:SOD}(4), there exists an arc $a_7\in\mathcal S$ that intersects $G$ outside of $a_1$ and $a_2$, for otherwise the great semicircle $G\setminus C$ would have no intersections with $\mathcal S$. Also, we may assume that $a_7$ is not contained in $G$ (if it is, pick any arc blocked by $a_7$ instead). Finally, let $a_8,a_9\in\mathcal S$ be two arcs blocked by $a_1$ and $a_2$, respectively.

Note that the seven arcs $a_3, a_4, \dots, a_9$ are all distinct and no great semicircle's relative interior completely contains more than one of them. This is because they all touch the great circle $G$ and none of them lies in $G$ (recall that each arc of a SOD is shorter than a great semicircle, due to axiom~A1). Thus, seven great semicircles are required to cover these arcs, and one additional great semicircle is required for $a_1$ and $a_2$. Therefore, $\mathcal S$ is good.
\end{proof}

\begin{corollary}\label{cor:final}
Every point that does not see any vertex of a polygonal scene $\mathcal P$ sees positive portions of at least eight distinct edges of polygons in  $\mathcal P$.
\end{corollary}
\begin{proof}
Let $p$ be any point that does not see any vertex of $\mathcal P$. Let $\mathcal S$ be the SOD generated by $\mathcal P$ with viewpoint $p$; for each edge $e$ of $\mathcal P$ that is visible to $p$, consider the line $\ell_e$ containing $e$. The projection of $\ell_e$ onto the unit sphere centered at $p$ is a great semicircle $C_e$ whose relative interior contains all the arcs of $\mathcal S$ corresponding to $e$. Thus, the set $\mathcal C$ of all great semicircles $C_e$ constructed as above for all edges $e$ visible to $p$ is a semicircle cover of $\mathcal S$. Since two distinct great semicircles $C_e,C_{e'}\in \mathcal C$ must correspond to distinct edges $e$ and $e'$, we conclude that $\mathcal S$ contains arcs corresponding to least $|\mathcal C|$ distinct edges of $\mathcal P$, all of which are visible to $p$. Due to \cref{l:semi}, we have $|\mathcal C|\geq 8$, and so $p$ sees at least eight edges of $\mathcal P$. To prove that $p$ sees \emph{positive portions} of such edges, it is sufficient to observe that the arcs of a SOD have positive length, and therefore are projections of positive portions of edges of $\mathcal P$.
\end{proof}

\begin{figure}[h]
\centering
\includegraphics[width=0.8\textwidth]{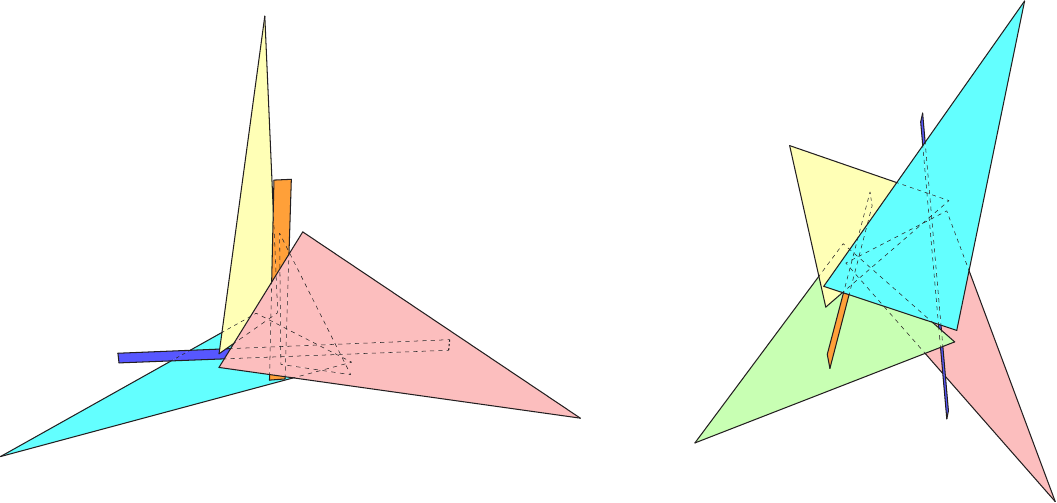}
\caption{Two views of a polygonal scene of six internally disjoint polygons where the central point sees no vertices and exactly eight edges}
\label{fig:7}
\end{figure}

\paragraph{Lower bound.}
We will now construct a polygonal scene $\mathcal P$ of six internally disjoint polygons matching the upper bound in \cref{cor:final}. The polygonal scene is depicted in \cref{fig:7}, as seen from two different viewpoints. As it will turn out, the visibility map generated by $\mathcal P$ with viewpoint the center of the arrangement is combinatorially equivalent to the SOD in \cref{fig:3}.

The polygons in $\mathcal{P}$ are as follows:
\begin{itemize}
\item a rectangle $R_1$ with vertex coordinates $(5, \pm 1, \pm 15)$,
\item a rectangle $R_2$ with vertex coordinates $(-5, \pm 15, \pm 1)$,
\item a triangle $T_1$ with vertex coordinates $(15, -2, 35)$, $(7, 0, -8)$, $(-7, -8, 3)$,
\item a triangle $T_2$ with vertex coordinates $(-15, -35, -2)$, $(-7, 8, 0)$, $(7, -3, -8)$,
\item a triangle $T_3$ with vertex coordinates $(15, 2, -35)$, $(7, 0, 8)$, $(-7, 8, -3)$,
\item a triangle $T_4$ with vertex coordinates $(-15, 35, 2)$, $(-7, -8, 0)$, $(7, 3, 8)$.
\end{itemize}

Let $\mathcal P=\{R_1,R_2,T_1,T_2,T_3,T_4\}$. Observe that the involution $\phi_1\colon (x,y,z)\mapsto (x,-y,-z)$ maps $\mathcal P$ to itself, because it fixes $R_1$ and $R_2$ and exchanges $T_1$ with $T_3$ and $T_2$ with $T_4$.

Similarly, the isometry $\phi_2\colon (x,y,z)\mapsto (-x,-z,y)$ has period $4$ and maps $\mathcal P$ to itself. In particular, $\phi_2$ exchanges $R_1$ and $R_2$ and maps $T_i$ to $T_{i+1}$ for all $1\leq i\leq 4$, where addition over the indices is taken modulo $4$.

\begin{proposition}\label{thm:lowerbound1}
The polygons in the scene $\mathcal P=\{R_1,R_2,T_1,T_2,T_3,T_4\}$ are disjoint.
\end{proposition}
\begin{proof}
The rectangles $R_1$ and $R_2$ are disjoint because they lie in distinct parallel planes: $x=5$ and $x=-5$, respectively.

The intersection between $T_1$ and the plane $x=5$ is a segment whose vertices have $y$-coordinate $-8/7$ and $-52/11$, respectively. Since both coordinates are smaller than $-1$, it follows that $R_1$ and $T_1$ are disjoint.

The intersection between $T_2$ and the plane $x=5$ is a segment whose vertices have $y$-coordinate $-10/7$ and $-65/11$. Both $y$-coordinates are smaller than $-1$, and therefore $R_1$ and $T_2$ are disjoint.

The orthogonal projections of $T_1$ and $T_2$ on the plane $y=0$ are two triangles that intersect only at their common vertex $(7,0,-8)$. Since this is a vertex of $T_1$ but not a vertex of $T_2$, it follows that $T_1$ and $T_2$ are disjoint.

Observe that all the points of $T_1$ have a negative $y$-coordinate except the vertex $(7,0,-8)$. On the other hand, the points of $T_3$ have a positive $y$-coordinate except the vertex $(7,0,8)$. Since the two vertices are distinct, the triangles $T_1$ and $T_3$ are disjoint.

Due to the symmetries of $\mathcal P$ given by $\phi_1$ and $\phi_2$, all other pairs of polygons in $\mathcal P$ are disjoint, as well.
\end{proof}

\begin{proposition}\label{thm:lowerbound2}
The point $p=(0,0,0)$ does not see any vertices of $\mathcal P=\{R_1,R_2,T_1,T_2,T_3,T_4\}$ and sees exactly eight of its edges.
\end{proposition}
\begin{proof}
Let $\alpha_i$ be the plane containing $T_i$, for $1\leq i\leq 4$. We will first show that the two vertices of $R_1$ with negative $z$-coordinate are occluded by $T_2$. Note that the ray $(5t,t,-15t)$ hits the vertex $(5,1,-15)$ of $R_1$ for $t=1$, and intersects the plane $\alpha_2$ with equation $7 x - 2 y + 15 z + 65 = 0$ for $t=65/192<1$ in the point $a_1=(325/192, 65/192, -325/64)$. Also, the ray $(5t,-t,-15t)$ hits the vertex $(5,-1,-15)$ of $R_1$ for $t=1$, and intersects $\alpha_2$ for $t=65/188<1$ in the point $a_2=(325/188, -65/188, -975/188)$. Observe that $a_1$ and $a_2$ can be expressed as a convex combination of the vertices of $T_2$ as follows:
\[
    a_1= \left(\frac{325}{192}, \frac{65}{192}, -\frac{325}{64}\right) = \frac{149}{8832}\cdot (-15, -35, -2) + \frac{519}{1472}\cdot (-7, 8, 0) + \frac{5569}{8832}\cdot (7, -3, -8),
\]
\[
    a_2= \left(\frac{325}{188}, -\frac{65}{188}, -\frac{975}{188}\right)=\frac{261}{8648}\cdot (-15, -35, -2) + \frac{1423}{4324}\cdot (-7, 8, 0) + \frac{5541}{8648}\cdot (7, -3, -8).
\]
Thus, $a_1$ and $a_2$ are in $T_2$, which implies that $T_2$ occludes both vertices of $R_1$ whose $z$-coordinate is negative. Since $T_2$ is convex, it occludes the entire edge of $R_1$ connecting these two vertices.

By the symmetry of $\mathcal P$ given by $\phi_1$, the edge of $R_1$ whose vertices have positive $z$-coordinate is occluded by $T_4$. Thus, all vertices of $R_1$ are occluded, and $p$ sees at most two edges of $R_1$. Due to the symmetry given by $\phi_2$, the same is true of $R_2$.

We will now show that $R_1$ occludes the two vertices of $T_1$ with positive $x$-coordinate. The plane $\beta$ containing $R_1$ has equation $x=5$; the ray $(15t,-2t,35t)$ hits the vertex $(15,-2,35)$ of $T_1$ for $t=1$ and intersects $\beta$ for $t=1/3<1$ in the point $b_1=(5,-2/3,35/3)$. The ray $(7t,0,-8t)$ hits the vertex $(7,0,-8)$ of $T_1$ for $t=1$ and intersects $\beta$ for $t=5/7<1$ in the point $b_2=(5,0,-40/7)$.  Since both $b_1$ and $b_2$ are in $R_1$, the two vertices of $T_1$ with positive $x$-coordinate are occluded. The edge of $T_1$ connecting them is also occluded, because $R_1$ is convex.

Finally, we will show that $T_4$ occludes the two vertices of $T_1$ with positive $z$-coordinate. The plane of $T_4$ is $\alpha_4$, with equation $7 x + 2 y - 15 z + 65 = 0$. The ray $(15t, -2t, 35t)$ hits the vertex $(15, -2, 35)$ of $T_1$ for $t=1$ and intersects $\alpha_4$ for $t=65/424<1$ in the point $c_1=(975/424, -65/212, 2275/424)$. The ray $(-7t, -8t, 3t)$ hits the vertex $(-7, -8, 3)$ of $T_1$ for $t=1$ and intersects $\alpha_4$ for $t=13/22<1$ in the point $c_2=(-91/22, -52/11, 39/22)$. Both $c_1$ and $c_2$ can be expressed as a convex combination of the vertices of $T_4$ as follows:
\[
    c_1= \left(\frac{975}{424}, -\frac{65}{212}, \frac{2275}{424}\right) = \frac{153}{19504}\cdot (-15, 35, 2) + \frac{1577}{4876}\cdot (-7, -8, 0) + \frac{13043}{19504}\cdot (7, 3, 8),
\]
\[
    c_2= \left(-\frac{91}{22}, -\frac{52}{11}, \frac{39}{22}\right)=\frac{21}{1012}\cdot (-15, 35, 2) + \frac{193}{253}\cdot (-7, -8, 0) + \frac{219}{1012}\cdot (7, 3, 8).
\]
Therefore, $c_1$ and $c_2$ are in $T_4$, implying that $T_4$ occludes both vertices of $T_1$ whose $z$-coordinate is positive. The edge of $T_1$ connecting them is also occluded, because $R_1$ is convex.

We conclude that all vertices and two edges of $T_1$ are occluded, and the same is true of $T_2$, $T_3$, and $T_4$, due to $\phi_2$. In total, the edges visible to $p$ are at most eight: at most two edges of $R_1$ and $R_2$, respectively, and at most one edge of $T_1$, $T_2$, $T_3$, and $T_4$, respectively. Thus, by \cref{cor:final}, exactly eight edges of $\mathcal P$ are visible to $p$.
\end{proof}

This completes the proof of \cref{thm:main2}.
To prove \cref{thm:main1}, it remains to show that there exists a polyhedron with the same visibility map as $\mathcal P=\{R_1,R_2,T_1,T_2,T_3,T_4\}$. This is achieved by a general procedure outlined in~\cite{vigliettaSOD}.
The following lemma shows how to augment $\mathcal{P}$ to a polyhedron $\mathcal{Q}$ with the same visibility map.
\begin{lemma}
Let $\mathcal{P}$ be a polygonal scene in $\mathbb{R}^3$ that consists of pairwise disjoint polygons, and let $p\in \mathbb{R}^3$ be a point in the convex hull of $\mathcal{P}$ that is disjoint from all polygons in $\mathcal{P}$. Then we can augment $\mathcal{P}$ to a polygonal scene $\mathcal{P}'$ such that $\mathcal{P'}$ forms the boundary of a polyhedron $\mathcal{Q}$, and $p$ has the same visibility map with respect to $\mathcal{P}$ and $\mathcal{Q}$.
\end{lemma}
\begin{proof}
Let $S$ be a sphere centered at $p$ that contains all polygons in $\mathcal{P}$. Consider the visibility map $\mathcal{A}$ of $p$ (generated by $\mathcal P$) on the sphere $S$. Note that $\mathcal{A}$ is an arrangement of noncrossing arcs of great circles. For example, if the polygonal scene $\mathcal P$ is the one in \cref{fig:7} and $p=(0,0,0)$, then, the visibility map $\mathcal{A}$ is (combinatorially equivalent to) the SOD in \cref{fig:3}.

Every face $F$ of $\mathcal{A}$ corresponds to a cone $C(F)$ with apex $p$, representing visibility rays emanating from $p$. The cone $C(F)$ either intersects a single polygon in $\mathcal{P}$ or does not intersect any polygon in $\mathcal{P}$ (i.e., for any point $q\in F$, the ray $\overrightarrow{pq}$ does not hit the interior of any polygon in $\mathcal{P}$). Let $\mathcal{F}$ be the set of faces of $\mathcal{A}$ of the second type, i.e., where $p$ does not see any polygon in $\mathcal{P}$. Since $p$ lies in the convex hull of $\mathcal{P}$, each face $F\in \mathcal{F}$ is strictly contained in some hemisphere $H(F)$. For every face $F\in \mathcal{F}$, successively create a new polygon $P(F)$ tangent to the sphere $S$ such that its orthographic projection to $S$ is contained in the hemisphere $H(F)$ and is disjoint from all other faces in $\mathcal{F}$.
Note that the polygonal scenes $\mathcal{P}$ and $\mathcal{P}\cup \{P(F): F\in \mathcal{F}\}$ have the same visibility map from $p$.

Let $S'$ be a sphere centered at $p$ that contains both $\mathcal{P}$ and polygons $P(F)$ for all $F\in \mathcal{F}$. We attach a ``funnel'' to each polygon $P$ in $\mathcal{P}\cup \{P(F): F\in \mathcal{F}\}$, on the opposite side from $p$, from the polygon $P$ to a pairwise disjoint ``holes'' in the sphere $S'$. To complete $\mathcal P'$, we connect the disjoint ``holes'' in $S'$ by a sufficiently fine mesh of $S'$ such that
the triangles are interior-disjoint from the funnels. Then the mesh and the funnels will bound a contractible polyhedron $\mathcal{Q}$. Since $p$ does not see any interior points of the polygons in the funnels and the mesh, the visibility map of $\mathcal{Q}$ is the same as that of $\mathcal{P}$.
\end{proof}

\section{Conclusions and Open Problems}\label{sec:conclusion}

We have proved that, given a polyhedron $\mathcal P$ in $\mathbb{R}^3$, any point $p\in \mathbb{R}^3$ sees at least six edges of $\mathcal P$. Moreover, if $p$ does not see any vertex of $\mathcal P$, it sees at least eight of its edges. Both bounds are tight, and the second one holds more generally for a polygonal scene $\mathcal P$ in $\mathbb{R}^3$. En route to these results, we also proved that any Spherical Occlusion Diagram (SOD) has at least four swirls and eight edges, and these bounds are tight, as well.

A possible direction for future research is extending our results to the number of visible \emph{faces} of polyhedra. Recall that, due to the Lusternik--Schnirelmann theorem, if a sphere is covered by three closed sets, one of the sets contains two antipodal points~\cite[pp.~118--119]{lust}. Applying this theorem to visibility maps, we immediately conclude that any interior point of a polyhedron sees at least four faces (because the projection of a face onto a sphere contains no antipodal points); a matching lower bound is trivially given by a tetrahedron. We conjecture that, if an interior point sees no vertices of a polyhedron, then it sees at least eight of its faces. A configuration matching this bound can be easily constructed from the arrangement in \cref{fig:7}.

Higher-dimensional generalizations of this problem can also be investigated: For example, given a polyhedron in $\mathbb{R}^d$, what is the minimum number of $\ell$-faces visible to a point that does not see any $k$-faces for $0\leq k<\ell<d$?

We believe that SODs are interesting objects in their own right, and will find more applications in discrete and computational geometry. Some open problems related to SODs are given in~\cite{vigliettaSOD}. Perhaps the most compelling issue is to find a simple characterization of the SODs that are not visibility maps of polyhedra. Such SODs are known to exist (see~\cite{KT22}), although they appear to be rare. An intriguing question is whether every SOD is \emph{combinatorially equivalent} to a visibility map of a polyhedron.

According to \cref{thm:swirls2}, any SOD has at least four swirls. Does it necessarily have two clockwise and two counterclockwise swirls? More generally, we ask for a characterization of the swirl graphs of SODs, as well as the swirl graphs of classes of SODs with certain properties (e.g., \emph{swirling SODs}, \emph{uniform SODs}, \emph{irreducible SODs}, and more~\cite{vigliettaSOD}).

\section*{Acknowledgments}
The authors are grateful to Joseph O'Rourke for insightful comments and suggestions that considerably improved the readability of this paper.
Research by T\'oth was partially supported by NSF DMS-1800734.
Research by Urrutia was partially supported by PAPIIT IN105221, Programa de Apoyo a la Investigaci\'on e Innovaci\'on Tecnol\'ogica UNAM.

\bibliographystyle{plainurl}
\bibliography{references}

\end{document}